\documentclass{article}%
\usepackage{amsmath}
\usepackage{amsfonts}
\usepackage{amssymb}
\usepackage{graphicx}%
\newtheorem{theorem}{Theorem}
\newtheorem{lemma}{Lemma}
\newtheorem{proposition}{Proposition}
\newtheorem{remark}{Remark}

\newcommand{\Cu}{c_1}
\newcommand{\Cd}{c_2}
\newcommand{\Ct}{c_3}
\newcommand{\Cq}{c_4}
\newcommand{\Cc}{c_5}
\newcommand{\Cs}{c_6}

\newcommand{\E}{\mathbb{E}}
\newcommand{\V}{\mathbb{V}}

\newcommand{\PP}{\mathbb{P}}

\newenvironment{proof}[1][Proof]{\noindent\textbf{#1.} }{\ \rule{0.5em}{0.5em}}

\title{Frontier estimation with local polynomials and high power-transformed data}

\author{S\'ephane Girard$^{(1)}$ \& Pierre Jacob$^{(2)}$}
\date{\small
$^{(1)}$team Mistis, INRIA Rh\^one-Alpes \& LJK, Inovall\'ee, 655, av. de l'Europe, Montbonnot, 38334 Saint-Ismier cedex, France, {\tt Stephane.Girard@inrialpes.fr}\\
(corresponding author)\\
$^{(2)}$Universit\'e Montpellier 2, EPS-I3M, place Eug\`ene Bataillon,\\
 34095 Montpellier cedex 5, France, {\tt jacob@math.univ-montp2.fr}\\
}

\begin{document}
\maketitle
\begin{abstract}
We present a new method for estimating the frontier of a sample.  
The estimator is based on a local polynomial regression on the
power-transformed data. We assume that the exponent of the transformation
goes to infinity while the bandwidth goes to zero.  
We give conditions on these two parameters to obtain
almost complete convergence. The asymptotic conditional bias and variance of
the estimator are provided and its good performance is
illustrated on some finite sample situations.\\



\end{abstract}

\section{Introduction}

Let $\left(  X_{i},Y_{i}\right)  $, $i=1,...,n$ be independent and identically
distributed continuous variables and suppose that their common density 
has a support defined by 
$$
S=\left\{  \left(  x,y\right)  \in\mathbb{R} \times\mathbb{R};0\leq y\leq g\left(  x\right)  \right\}.
$$
The unknown
function $g$ is called the \textit{frontier}. We address the problem of
estimating $g$. In~\cite{JMVA2}, we introduced a new kind of estimator
based upon kernel regression on high power-transformed data. More precisely
the estimator of $g(x)$ was defined by
\[
\left(  \left(  p+1\right)  \sum_{i=1}^{n}K_{h}\left(  X_{i}-x\right)
Y_{i}^{p}\left/\sum_{i=1}^{n}K_{h}\left(  X_{i}-x\right) \right. \right)  ^{1/p}%
\]
where $p=p_{n}\rightarrow\infty$ and $h=h_{n}\rightarrow0$ are non
random sequences, $K$ is a symmetrical probability density with support
included in $\left[  -1,1\right]$, and $K_{h}(\bullet)=K\left(  \bullet/h\right)
/h$. Although the correcting term $\left(  p+1\right)  ^{1/p}$ was
specially designed to deal with the case of a uniform conditional distribution
of $Y/X=x$, this estimate has been shown to converge in any case. In the
special but interesting case of a uniform conditional distribution of $Y/X=x$
for a $\alpha-$lipschitzian frontier the minimax rate of convergence is
attained. We also proved that the estimator is asymptotically Gaussian.
It is also interesting to note that, compared to the
extreme value based 
estimators~\cite{ISUPLaurent,Geffroy,Scandi,JSPI,JSPI2,ESAIM}, 
projection estimators~\cite{JacSuq}
or piecewise polynomial estimators~\cite{KorTsy,KorTsy3,Har},
this estimator does not require a partition of the support $S$.

A natural idea suggested by our referees was to investigate the possible
gains obtained by substituting a local polynomial regression to the
Nadaraya-Watson regression. The basic idea in this theory
consists in approximating locally a $C_{k+1}$ regression function by a
polynomial of degree $k$ and taking the zero-degree term as an estimate of the
regression. The regularity of the function brings improvement on the
bias term.
Accordingly, when dealing with high power-transformed data we establish in this paper that the
bias of the local polynomial estimator of degree $k$ is $O_{p}(h(hp)^{k})$
and the variance is $O_{p}\left(  1/nhp\right)  $.

Let us introduce the notations $Z=\left(  p+1\right)  Y^{p}$ and $r_{n}\left(
x\right)  =\E\left(  Z/X=x\right)  $. The conditional distribution of $Y/X=x$
is supposed to be uniform on $\left[  0,g\left(  x\right)  \right]  $, so that
$r_{n}\left(  x\right)  =g^{p}\left(  x\right)  $. For fixed $p$ the method
for estimating $r_{n}\left(  x\right)  $ first consists in solving the
following minimization problem%
\begin{equation}
\underset{ \beta_{0},...,\beta_{k}}%
{\arg\min}\sum_{i=1}^{n}\left(  \left(  p+1\right)  Y_{i}^{p}-\sum_{j=0}%
^{k}\beta_{j}\left(  X_{i}-x\right)  ^{j}\right)  ^{2}K_{h}\left(
X_{i}-x\right)  \text{.} \label{1-0}%
\end{equation}
Then, denoting by 
$\widehat{\mathbf{\beta}}=(  \widehat{\beta}%
_{0},...,\widehat{\beta}_{k})  ^{t}$ the solution of this least square
minimization, one considers $\widehat{\beta}_{0}$ as an estimate of
$r_{n}\left(  x\right)  =\E\left(  Z/X=x\right)  $. The originality and the
difficulty of our paper in contrast with these traditional lines is that here
$p=p_{n}\rightarrow\infty$ and that we consider $\widehat{\beta}_{0}^{1/p}$ as
an estimate of $g\left(  x\right)  .$ So we write $\widehat{g}_{n}\left(
x\right)  =\widehat{\beta}_{0}^{1/p}=\widehat{r}_{n}^{1/p}\left(  x\right)  $.
We refer to~\cite{Hall,Hall3,Keith} for other definitions of local
polynomials estimators (i.e. without high power transform)
and to~\cite{DST,Farrel,Gijbels2,Aragon,Cazals} for the estimation of frontier functions under monotonicity assumptions.  

In order to get simplified matricial expressions, let us denote by $\mathbf{X}$ the
$n\times\left(  k+1\right)  $ matrix defined by the lines
 $[1,X_{i}-x,...,\left(  X_{i}-x\right)  ^{k}] _{i=1,...n}$. The
diagonal matrix of weights $diag\left\{  K_{h}\left(  X_{i}-x\right)
\right\}  $ is denoted by $\mathbf{W}$. We call \textit{design} the
vector$\ \mathcal{X}=(X_{1},...,X_{n})^{t}$ and we denote by $\mathcal{Z}$ the
vector $\left(  Z_{1},...,Z_{n}\right)^{t}$. Then the local regression
problem $\left(  \ref{1-0}\right)  $ can be rewritten as

\[
\widehat{\mathbf{\beta}}=\arg\min_{\beta}\left(  \mathcal{Z}\mathbf{-X\beta
}\right)  ^{t}\mathbf{W}\left(  \mathcal{Z}\mathbf{-X\beta}\right)  ,
\]
where
$\beta=(\beta_0,\dots,\beta_k)^t$.
It is well known from the weighted least square theory that
\[
\widehat{\mathbf{\beta}}=\left(  \mathbf{X}^{t}\mathbf{WX}\right)
^{-1}\mathbf{X}^{t}\mathbf{W}\mathcal{Z}.
\]
In particular, in the case $k=0$ we have
\[
\widehat{\mathbf{\beta}}=\widehat{\beta}_{0}=\sum_{i=1}^{n}Z_{i}K_{h}\left(
X_{i}-x\right) \left/\sum_{i=1}^{n}K_{h}\left(  X_{i}-x\right) \right. ,
\]
so we exactly find back the estimator $\widehat{g}_{n}\left(  x\right)
=\widehat{\beta}_{0}^{1/p}$ studied in~\cite{JMVA2}.
In order to give a general expression of $\widehat{r}_{n}\left(  x\right)$,
we adopt the notations of
Fan and Gijbels whose book~\cite{FanGij} will also serve of
reference for some preliminary results established in Section~\ref{preli}
(see also~\cite{RupWan} for a general multidimensional analysis).
Basing on this, the asymptotic conditional bias and variance
of the estimator are derived in Section~\ref{condi}
when $Y$ given $X=x$ is uniformly distributed.
This result is extended in Section~\ref{comple}, where the almost
complete convergence is proved without this uniformity assumption.
We conclude this paper by an illustration of the behavior
of our estimator on some finite sample situations in Section~\ref{simul}.
Technical lemmas are postponed to the appendix.

\section{Preliminary results}
\label{preli}

Let $x\in {\mathbb{R}}$. From now on, it is assumed that
the density function $f$ of $X_1$ is continuous at $x$ and that $f(x)>0$.
Besides, we suppose that there exists $g_{\min}>0$
such that, for all $t\in{\mathbb{R}}$, $g_{\min}\leq g(t)$.
Let $\mathbf{S}_{n}=\mathbf{X}^{t}\mathbf{WX}$ be the $\left(  k+1\right)
\times\left(  k+1\right)  $ matrix $\left[  S_{n,j+l}\right]  $ $_{0\leq
j,l\leq k}$ defined by 
$$
S_{n,j}=\sum_{i=1}^{n}\left(  X_{i}-x\right)  ^{j} K_{h}\left(  X_{i}-x\right).
$$
 Similarly, denoting by $\mathbf{\Sigma}$ the
$n\times n$ diagonal matrix $diag\left\{  K_{h}^{2}\left(  X_{i}-x\right)
g^{2p}\left(  X_{i}\right)  \right\}  $, $\mathbf{S}_{n}^{\ast}~=\mathbf{X}%
^{t}\mathbf{\Sigma X}$ is the $\left(  k+1\right)  \times\left(  k+1\right)  $
matrix\ $[  S_{n,j+l}^{\ast}]_{0\leq j,l\leq k}$ with
$$
S_{n,j}^{\ast}=\sum_{i=1}^{n}\left(  X_{i}-x\right)  ^{j}K_{h}^{2}\left(
X_{i}-x\right)  g^{2p}\left(  X_{i}\right).
$$
Finally, we introduce the
matrices $\mathbf{S}=\left[  \mu_{j+l}\right]  _{0\leq j,l\leq k}$ and
$\mathbf{S}^{\ast}=$ $\left[  \nu_{j+l}\right]  _{0\leq j,l\leq k}$ with
$\mu_{j}=\int u^{j}K\left(  u\right)  du$ and $\nu_{j}=\int u^{j}K^{2}\left(
u\right)  du$.
Following roughly the same lines as Fan and Gijbels~\cite{FanGij},
we obtain asymptotic expressions for $S_{n,j}$ and $S_{n,j}^{\ast}$. The
first equality~(\ref{1-3}) is a standard result of the
theory and the second one~(\ref{1.4}) boils down to an easy
adaptation. Proofs are thus omitted.

\begin{proposition}
\label{prop1}
If $h\rightarrow0$ and $nh\rightarrow\infty$, then 
\begin{equation}
S_{n,j}=nh^{j}f\left(  x\right)  \mu_{j}\left[  1+o_{p}\left(  1\right)
\right]  . \label{1-3}%
\end{equation}
If, moreover, $ph\rightarrow0$ we have for any $C_{1}$ function $g$
\begin{equation}
S_{n,j}^{\ast}=nh^{j-1}g^{2p}\left(  x\right)  f\left(  x\right)  \nu
_{j}\left[  1+o_{p}\left(  1\right)  \right]  \text{.} \label{1.4}%
\end{equation}
\end{proposition}
Let us now quote a general expression of the conditional bias of$\ \widehat{r}_{n}\left(x\right)$. 
From Fan and Gijbels~\cite{FanGij}, 
and denoting by $e_{1}=(1,0,...,0)^{t}$ the first vector of the canonical basis of
$\mathbb{R}^{k+1}$, we have
\begin{eqnarray*}
r_{n}\left(  x\right)  &=&\beta_{0}=e_{1}^{t}\beta\mathbf{~=}~e_{1}%
^{t}\mathbf{S}_{n}^{-1}\mathbf{X}^{t}\mathbf{WX\beta,}\\
\widehat{r}_{n}\left(  x\right)  &=&\widehat{\beta}_{0}=e_{1}^{t}\widehat
{\mathbf{\beta}}=e_{1}^{t}\mathbf{S}_{n}^{-1}\mathbf{X}^{t}\mathbf{W}%
\mathcal{Z},
\end{eqnarray*}
so that
\begin{equation}
\E\left(  \widehat{r}_{n}\left(  x\right)  /\mathcal{X}\right)  -r_{n}\left(
x\right)  =e_{1}^{t}\mathbf{S}_{n}^{-1}\mathbf{X}^{t}\mathbf{W}\left[
\E\left(  \mathcal{Z}/\mathcal{X}\right)  -\mathbf{X\beta}\right]  .
\label{1-1}%
\end{equation}
In Appendix I we give a detailed proof of the following
\begin{proposition}
\label{prop2}
Suppose $g$ is a $C_{k+1}$ function. If $h\rightarrow0$,
$nh\rightarrow\infty$ and $ph\rightarrow0$, then 
$$
\E\left(  \frac{\widehat
{r}_{n}\left(  x\right)  }{r_{n}\left(  x\right)  }-1/\mathcal{X}\right)
=O_{p}\left(  \left(  hp\right)  ^{k+1}\right).
$$
\end{proposition}
We now examine the conditional variance of
$\widehat{r}_{n}\left(  x\right)  $%
\[
\V\left(  \widehat{r}_{n}\left(  x\right)  /\mathcal{X}\right)  =e_{1}%
^{t}\mathbf{S}_{n}^{-1}\mathbf{X}^{t}\V\mathbf{\left(  \mathbf{W}%
\mathcal{Z}\mathbf{/}\mathcal{X}\right)  XS}_{n}^{-1}e_{1}.
\]
Taking into account of the independence of the pairs $\left(  X_{i},Y_{i}\right)
,\V\mathbf{\left(  \mathbf{W}\mathcal{Z}\mathbf{/}\mathcal{X}\right)  }$ is
the diagonal matrix $diag\left\{  K_{h}^{2}\left(  X_{i}-x\right)
\V(Z_{i}/X_{i}=x)\right\}  $. From the uniformity of the conditional
distribution of the $Y_{i}/X_{i}=x$, it is easily seen that $\V(Z_{i}%
/X_{i}=x)=\frac{p^{2}}{2p+1}g^{2p}\left(  x\right)  $, so that
\[
\V\left(  \widehat{r}_{n}\left(  x\right)  /\mathcal{X}\right)  =\frac{p^{2}%
}{2p+1}e_{1}^{t}\mathbf{S}_{n}^{-1}\mathbf{S}_{n}^{\ast}\mathbf{S}_{n}%
^{-1}e_{1}.
\]
Following the same lines as Fan and Gijbels~\cite{FanGij}, we
obtain the following asymptotic expression
\begin{proposition}
\label{prop3}
Suppose $g$ is a $C_{k+1}$ function. If $h\rightarrow0$,
$nh\rightarrow\infty$ and $ph\rightarrow0$, then 
\[
\V\left(  \frac{\widehat{r}_{n}\left(  x\right)  }{r_{n}\left(  x\right)
}/\mathcal{X}\right)  =\frac{C}{f\left(  x\right)  }\frac{1}{nh}\frac{p^{2}%
}{2p+1}\left[  1+o_{p}\left(  1\right)  \right],
\]
where $C=e_{1}^{t}\mathbf{S}^{^{-1}}\mathbf{S}^{\ast}\mathbf{S}^{^{-1}}e_{1}$.
\end{proposition}
The proof of Proposition~\ref{prop3} is much easier than the one of 
Proposition~\ref{prop2} and it thus omitted.

\section{Conditional bias and variance of $\widehat{g}_{n}\left(  x\right)  $}
\label{condi}

Here we present the main results of this paper and an outline of their
proofs. Many details and ancillary results are postponed to Appendix II.
Proofs are made under the assumption that $g$ is a $C_{k+1}$ function
and the system of conditions below
\[
H:%
\begin{bmatrix}
n\rightarrow\infty,~h\rightarrow0,~p\rightarrow\infty\\
nh\rightarrow\infty,~hp\rightarrow0\\
\left(  p/nh\right)  \log^{2}(nh)\sim\left(  hp\right)  ^{2k+2}%
\end{bmatrix}
.
\]
\begin{theorem}
\label{th4}
Suppose $H$ holds and $g$ is a $C_{k+1}$ function. Then, the asymptotic
conditional bias of the estimate is given by 
$$
\E\left(  \frac{\widehat{g}_{n}\left(  x\right)  }{g\left(  x\right)  }-1/\mathcal{X}\right)
=O_{p}\left(  h\left(  hp\right)  ^{k}\right) .
$$
\end{theorem}

\begin{proof}
Let us write $w_{n}\left(  x\right)  =\widehat{r}_{n}\left(  x\right)
/r_{n}\left(  x\right)  -1$, so that $\widehat{g}_{n}\left(  x\right)
/g\left(  x\right)  =\left(  1+w_{n}\left(  x\right)  \right)  ^{1/p}$ and 
define%
\begin{equation}
\label{defdelta}
\Delta_n=\left\vert \left(  1+w_{n}\left(  x\right)  \right)  ^{1/p}-\left(
1+\frac{w_{n}\left(  x\right)  }{p}\right)  \right\vert .
\end{equation}
Let $\alpha_{n}=\left(  p/nh\right)  ^{1/4}$. For sufficiently large $n$ we
have $\alpha_{n}<1/2$, and thus, Lemma~\ref{lem16} entails
\begin{equation}
\Delta_n\mathbf{1}\left\{  \left\vert w_{n}\left(  x\right)  \right\vert
<\alpha_{n}\right\}  <\Cs\frac{1}{p}w_{n}^{2}\left(  x\right)  \label{2-0},
\end{equation}
which leads to the following bound%
\[
\E\left(  \Delta_n\mathbf{1}\left\{  \left\vert w_{n}\left(  x\right)
\right\vert <\alpha_{n}\right\}  /\mathcal{X}\right)  \leq \Cs\frac{\alpha_{n}%
}{p}\E\left(  \left\vert w_{n}\left(  x\right)  \right\vert /\mathcal{X}%
\right)  \leq \Cs\frac{\alpha_{n}}{p}\E^{1/2}\left(  w_{n}^{2}\left(  x\right)
/\mathcal{X}\right).
\]
Now, from Proposition~\ref{prop2} and Proposition~\ref{prop3},
\begin{align*}
\E\left(  w_{n}^{2}\left(  x\right)  /\mathcal{X}\right)   &  =\V\left(
\frac{\widehat{r}_{n}\left(  x\right)  }{r_{n}\left(  x\right)  }%
/\mathcal{X}\right)  +\E^{2}\left(  \frac{\widehat{r}_{n}\left(  x\right)
}{r_{n}\left(  x\right)  }-1/\mathcal{X}\right)  \\
&  =\frac{\Cs}{f\left(  x\right)  }\frac{1}{nh}\frac{p^{2}}{2p+1}\left[
1+o_{p}\left(  1\right)  \right]  +O_{p}\left(  \left(  hp\right)
^{2k+2}\right)  \text{.}%
\end{align*}
Then, taking into account of $h\left(  hp\right)  ^{k}\sqrt{nhp}=\sqrt{\log
(nh)}\rightarrow\infty$, it follows that 
\begin{align}
\E\left(  \Delta_n\mathbf{1}\left\{  \left\vert w_{n}\left(  x\right)
\right\vert <\alpha_{n}\right\}  /\mathcal{X}\right)   &  \leq \Cs\frac
{\alpha_{n}}{p}\left[  O_{p}\left(  \frac{p}{nh}\right)  +O_{p}\left(  \left(
hp\right)  ^{2k+2}\right)  \right]  ^{1/2}\nonumber\\
&  =O_{p}\left(  \alpha_{n}/\sqrt{nhp}\right)  +O_{p}\left(  \alpha
_{n}h\left(  hp\right)  ^{k}\right) \nonumber\\
& =O_{p}\left(  \alpha_{n}h\left(
hp\right)  ^{k}\right).  \label{2-1}%
\end{align}
Besides, making use of Lemma~\ref{lem15}, we can write%
\[
\E\left(  \Delta_n\mathbf{1}\left\{  \left\vert w_{n}\left(  x\right)
\right\vert \geq\alpha_{n}\right\}  /\mathcal{X}\right)  \leq \Cc\left(
\mathcal{X}\right)  \PP\left\{  \left\vert w_{n}\left(  x\right)  \right\vert
\geq\alpha_{n}/\mathcal{X}\right\},
\]
and, from the triangular inequality,
\begin{align*}
\PP\left\{  \left\vert w_{n}\left(  x\right)  \right\vert \geq\alpha
_{n}/\mathcal{X}\right\}   &  \leq \PP\left\{  2\left\vert w_{n}\left(
x\right)  -\E\left(  w_{n}\left(  x\right)  /\mathcal{X}\right)  \right\vert
\geq\alpha_{n}/\mathcal{X}\right\}  \\
&  +\PP\left\{  2\left\vert \E\left(  w_{n}\left(  x\right)  /\mathcal{X}\right)
\right\vert \geq\alpha_{n}/\mathcal{X}\right\}  .
\end{align*}
Recalling that 
$$
\E\left(  w_{n}\left(  x\right)  /\mathcal{X}\right)  =\E\left(
\frac{\widehat{r}_{n}\left(  \mathcal{X}\right)  }{r_{n}\left(  \mathcal{X}%
\right)  }-1/\mathcal{X}\right)  =O_{p}\left(  \left(  hp\right)
^{k+1}\right),
$$
and noticing that $\left(  hp\right)  ^{k+1}/\alpha
_{n}=\left(  p/nh\right)  ^{1/4}\left(  \log(nh)\right)  ^{1/2}\rightarrow0$,
we conclude that the sequence $\PP\left\{  2\left\vert \E\left(  w_{n}\left(
x\right)  /\mathcal{X}\right)  \right\vert \geq\alpha_{n}/\mathcal{X}\right\}
$ goes to $0$. 
Moreover, remark that  $\PP\left\{  2\left\vert \E\left(  w_{n}\left(  x\right)
/\mathcal{X}\right)  \right\vert \geq\alpha_{n}/\mathcal{X}\right\}  $ is a
$\left\{  0,1\right\}  $-valued random variable. 
This means that for a
sufficient large $n$ depending on $\mathcal{X}$, we merely have 
$$
\PP\left\{
2\left\vert \E\left(  w_{n}\left(  x\right)  /\mathcal{X}\right)  \right\vert
\geq\alpha_{n}/\mathcal{X}\right\}  =0.
$$
Now, from Lemma~\ref{lem14},
\begin{align*}
&  \PP\left(  2\left\vert w_{n}\left(  x\right)  -\E\left(  w_{n}\left(
x\right)  /\mathcal{X}\right)  \right\vert \geq\alpha_{n}/\mathcal{X}\right)
\\
&  =\PP\left(  \left\vert \widehat{r}_{n}\left(  x\right)  -\E\left(  \widehat
{r}_{n}\left(  x\right)  /\mathcal{X}\right)  \right\vert \geq\frac{1}%
{2}\alpha_{n}r_{n}\left(  x\right)  /\mathcal{X}\right)  \\
&  \leq2\exp\left\{  -\Cq\frac{nh}{p}\frac{\alpha_{n}^{2}}{4}\left[
1+o_{p}\left(  1\right)  \right]  \right\}\\
&  =2\exp\left\{  -\frac{\Cq}{4}\sqrt
{nh/p\log^{2}(nh)}\left[  1+o_{p}\left(  1\right)  \right]  \log(nh)\right\}  \\
&  =\left(  nh\right)  ^{-\infty_{p}\left(  1\right)  },
\end{align*}
where $\infty_{p}\left(  1\right)  $ stands for a sequence going almost surely
to the infinity. We thus have at least%
\begin{equation}
\E\left(  \Delta_n\mathbf{1}\left\{  \left\vert w_{n}\left(  x\right)
\right\vert \geq\alpha_{n}\right\}  /\mathcal{X}\right)  =O_{p}\left(
1/nh\right).  \label{2-2}%
\end{equation}
Collecting $\left(  \ref{2-1}\right)$ and $\left(  \ref{2-2}\right)  $ yields
\[
\E\left(  \Delta_n/\mathcal{X}\right)  =O_{p}\left(  \alpha_{n}h\left(
hp\right)  ^{k}\right)  +O_{p}\left(  1/nh\right).
\]
From%
\[
\left\vert \E\left(  \frac{\widehat{g}_{n}\left(  x\right)  }{g\left(
x\right)  }-1/\mathcal{X}\right)  -\frac{1}{p}\E\left(  w_{n}\left(  x\right)
/\mathcal{X}\right)  \right\vert \leq \E\left(  \Delta_n/\mathcal{X}\right)
\]
and Proposition~\ref{prop2},  we obtain
\begin{align}
\E\left(  \frac{\widehat{g}_{n}\left(  x\right)  }{g\left(  x\right)
}-1/\mathcal{X}\right)   &  =\frac{1}{p}\E\left(  w_{n}\left(  x\right)
/\mathcal{X}\right)  +O_{p}\left(  \alpha_{n}h\left(  hp\right)  ^{k}\right)
+O_{p}\left(  1/nh\right)  \nonumber\\
&  =O_{p}\left(  h\left(  hp\right)  ^{k}\right)  +O_{p}\left(  1/nh\right)
.\label{2-3}%
\end{align}
Finally, since $h\left(  hp\right)  ^{k}nh=\left(  nh/p\right)  \left(
hp\right)  ^{k+1}=\sqrt{nh/p}\log(nh)\rightarrow\infty$, expansion~$\left(
\ref{2-3}\right)  $ reduces to%
\[
\E\left(  \frac{\widehat{g}_{n}\left(  x\right)  }{g\left(  x\right)
}-1/\mathcal{X}\right)  =O_{p}\left(  h\left(  hp\right)  ^{k}\right),
\]
and the conclusion follows.
\end{proof}

\begin{theorem}
Suppose $H$ holds and $g$ is a $C_{k+1}$ function. Then, the asymptotic
conditional variance of the estimate is given by 
$$
\V\left(  \frac
{\widehat{g}_{n}\left(  x\right)  }{g\left(  x\right)  }-1/\mathcal{X}%
\right)  =O_{p}\left(  \frac{1}{nhp}\right).
$$
\end{theorem}

\begin{proof}
Introducing%
\[
\delta=\frac{\widehat{g}_{n}\left(  x\right)  }{g\left(  x\right)
}-1-\frac{w_{n}\left(  x\right)  }{p},
\]
we have%
\[
\V\left(  \frac{\widehat{g}_{n}\left(  x\right)  }{g\left(  x\right)
}/\mathcal{X}\right)  \leq\frac{2}{p^{2}}\V\left(  w_{n}\left(  x\right)
/\mathcal{X}\right)  +2\V\left(  \delta/\mathcal{X}\right).
\]
The first term is bounded using Proposition~\ref{prop3}:
\[
\frac{1}{p^{2}}\V\left(  w_{n}\left(  x\right)  /\mathcal{X}\right)
=\frac{1}{p^{2}}\V\left(  \frac{\widehat{r}_{n}\left(  x\right)  }%
{r_{n}\left(  x\right)  }/\mathcal{X}\right)  =O_{p}\left(  \frac{1}%
{nhp}\right).
\]
Second,
\[
\V\left(  \delta/\mathcal{X}\right)  \leq \E\left(  \delta^{2}/\mathcal{X}%
\right)  =\E\left(  \Delta_n^{2}/\mathcal{X}\right),
\]
and (\ref{2-0}) yields, for sufficiently large $n$, 
\[
\Delta_n^{2}\mathbf{1}\left\{  \left\vert w_{n}\left(  x\right)  \right\vert
<\alpha_{n}\right\}  <\Cs^{2}\frac{1}{p^{2}}w_{n}^{4}\left(  x\right)
<\Cs^{2}\frac{\alpha_{n}^{2}}{p^{2}}w_{n}^{2}\left(  x\right),
\]
which entails%
\begin{align*}
\E\left(  \Delta_n^{2}\mathbf{1}\left\{  \left\vert w_{n}\left(  x\right)
\right\vert <\alpha_{n}\right\}  /\mathcal{X}\right)   &  \leq \Cs^{2}%
\frac{\alpha_{n}^{2}}{p^{2}}\E\left(  w_{n}^{2}\left(  x\right)  /\mathcal{X}%
\right) \\
&  =\frac{\alpha_{n}^{2}}{p^{2}}\left[  O_{p}\left(  \frac{p}{nh}\right)
+O_{p}\left(  \left(  hp\right)  ^{2k+2}\right)  \right]  .
\end{align*}
In a similar way as in the previous proof, one has
\begin{align*}
\E\left(  \Delta_n^{2}\mathbf{1}\left\{  \left\vert w_{n}\left(  x\right)
\right\vert \geq\alpha_{n}\right\}  /\mathcal{X}\right)   &  \leq \Cc\left(
{\mathcal{X}}\right)  \PP\left\{  \left\vert w_{n}\left(  x\right)  \right\vert \geq
\alpha_{n}/\mathcal{X}\right\} \\
&  =\left(  nh\right)  ^{-\infty_{p}\left(  1\right)  }=O_{p}\left(  \frac
{1}{n^{2}h^{2}}\right).
\end{align*}
It follows that
\[
\E\left(  \Delta_n^{2}/\mathcal{X}\right)  =O_{p}\left(  \frac{\alpha_{n}^{2}%
}{nhp}\right)  +O_{p}\left(  \alpha_{n}^{2}h^{2}\left(  hp\right)
^{2k}\right)  +O_{p}\left(  \frac{1}{n^{2}h^{2}}\right),
\]
and, taking account of $\alpha_{n}=\left(  p/nh\right)  ^{1/4}$ and
$nh/p\log^{2}(nh)\rightarrow\infty$,
we finally obtain
\begin{align*}
\V\left(  \frac{\widehat{g}_{n}\left(  x\right)  }{g\left(  x\right)
}/\mathcal{X}\right)   &  =O_{p}\left(  \frac{1}{nhp}\right)  +O_{p}\left(
\alpha_{n}^{2}h^{2}\left(  hp\right)  ^{2k}\right)  +O_{p}\left(  \frac
{1}{n^{2}h^{2}}\right) \\
&  =O_{p}\left(  \frac{1}{nhp}\right),
\end{align*}
and the result is proved.
\end{proof}

\begin{remark} Under the assumptions of the above theorems, 
the conditional mean square error is given by
\begin{align*}
\E\left[\left( \frac{\widehat{g}_{n}\left(  x\right)
}{g\left(  x\right)  }-1\right)^2/\mathcal{X} \right]
&=  \V\left(  \frac{\widehat{g}_{n}\left(  x\right)  }{g\left(  x\right)
}-1/\mathcal{X}\right)  +\E^{2}\left(  \frac{\widehat{g}_{n}\left(  x\right)
}{g\left(  x\right)  }-1/\mathcal{X}\right) \\
&  =O_{p}\left(  \frac{1}{nhp}\right)  +O_{p}\left(  h^{2}\left(  hp\right)
^{2k}\right) \\
&  =O_{p}\left(  h^{2}\left(  hp\right)  ^{2k}\right)  =O_{p}\left(  \frac
{1}{nhp}\log^{2}(nh)\right).
\end{align*}
\end{remark}
Under condition H, the ratio between the bias and variance terms
is asymptotically equivalent to $\log^2(nh)$. Thus, bias and variance of 
$\widehat{g}_{n}\left(  x\right)$ are approximatively of same order, up to this
logarithmic factor.

\section{Convergence of $\widehat{g}_{n}\left(  x\right)  $ under
general conditions}
\label{comple}

In this section, the almost complete convergence of  $\widehat{g}_{n}\left(  x\right)  $
is established without any assumption on the conditional distribution of $Y$ given $X$. 

\begin{theorem}
If  $h\rightarrow0$, $p\rightarrow\infty$, and $nh/\log n\rightarrow\infty$,
then $\widehat{g}_{n}\left(  x\right)  $ converges to $g\left(  x\right)  $
almost completely. 
\end{theorem}

\begin{proof}
Introducing 
\[
a\left(  X_{i}\right)  =\frac{1}{nf\left(  x\right)  }K_{h}\left(
X_{i}-x\right)  \sum_{j=0}^{k}u_{j}\left(  \frac{X_{i}-x}{h}\right)  ^{j}%
\]
and$\ \widehat{s}_{n}\left(  x\right)  =\sum\nolimits_{i=1}^{n}a\left(
X_{i}\right)  Z_{i}$, Lemma~\ref{lem9} entails that $\widehat{r}_{n}$
can be rewritten as
\[
\widehat{r}_{n}\left(  x\right)  =\widehat{s}_{n}\left(  x\right)
+\widehat{s}_{n}\left(  x\right)  o_{p}\left(  1\right).
\]
Thus, with $2\eta=\varepsilon/g\left(  x\right)  $ and since
$[1+o_{p}(1)]^{1/p}=[  1+o_{p}(1)]  $, we have
\begin{align*}
\left\{  \left\vert \widehat{g}_{n}\left(  x\right)  -g\left(  x\right)
\right\vert >\varepsilon\right\}   &  =\left\{  \left\vert \frac{\widehat
{r}_{n}^{1/p}\left(  x\right)  }{g\left(  x\right)  }-1\right\vert
>2\eta\right\} \\
&  \subseteq\left\{  \left\vert \frac{\widehat{s}_{n}^{1/p}\left(  x\right)  }{g\left(
x\right)  }-1\right\vert >\eta\right\}  \cup\left\{  \left\vert \frac
{\widehat{s}_{n}^{1/p}\left(  x\right)  }{g\left(  x\right)  }o_{p}\left(
1\right)  \right\vert >\eta\right\},
\end{align*}
with
\[
\frac{\widehat{s}_{n}^{1/p}\left(  x\right)  }{g\left(  x\right)  }=\left[
\sum_{i=1}^{n}a\left(  X_{i}\right)  \left(  p+1\right)  \left[
\frac{Y_{i}}{g\left(  x\right)  }\right]  ^{p}\right]  ^{1/p}.
\]
Since $\left(  1+p\right)  ^{1/p}\rightarrow1$, let us focus on
\[
T_{n}\left(  x\right)  =\left[  \sum_{i=1}^{n}a\left(  X_{i}\right)
\left[  \frac{Y_{i}}{g\left(  x\right)  }\right]  ^{p}\right]  ^{1/p}.
\]
Taking $0<\delta<\eta$, $\left\vert X_{i}-x\right\vert <h$ implies
$Y_{i}-g\left(  x\right)  \left(  1+\delta\right)  <0$ and thus
\begin{align*}
T_{n}\left(  x\right)   &  =\left[  \sum_{i=1}^{n}a\left(
X_{i}\right)  \left[  \frac{Y_{i}}{g\left(  x\right)  }\right]  ^{p}%
\mathbf{1}\left\{  Y_{i}<g\left(  x\right)  \left(  1+\delta\right)  \right\}
\right]  ^{1/p}\\
&  \leq\left(  1+\delta\right)  \left[  \sum_{i=1}^{n}a\left(
X_{i}\right)  \mathbf{1}\left\{  Y_{i}<g\left(  x\right)  \left(
1+\delta\right)  \right\}  \right]  ^{1/p}.
\end{align*}
Moreover, since, for $n$ large enough, $\left(  \frac{1+\eta}{1+\delta
}\right)  ^{p}>2$, it follows that
\begin{align*}
&\left\{  T_{n}\left(  x\right)  >1+\eta\right\}\\
   &  \subseteq\left\{
\sum_{i=1}^{n}a\left(  X_{i}\right)  \mathbf{1}\left\{
Y_{i}<g\left(  x\right)  \left(  1+\delta\right)  \right\}  >2\right\} \\
&  =\left\{  \frac{1}{n}\sum_{i=1}^{n}K_{h}\left(  X_{i}-x\right)
\sum_{j=0}^{k}u_{j}\left(  \frac{X_{i}-x}{h}\right)  ^{j}\mathbf{1}\left\{
Y_{i}<g\left(  x\right)  \left(  1+\delta\right)  \right\}  \frac{1}{f\left(
x\right)  }>2\right\}.
\end{align*}
Now, the only difference with the proof of Theorem~1 in~\cite{JMVA2} is
that the positive kernel $K\left(  x\right)  $ is replaced by the signed
kernel of higher order $K\left(  x\right)  \sum_{j=0}^{k}u_{j}x^{j}$. 
The case
$\left\{  T_{n}\left(  x\right)  <1-\eta\right\}  $ is easily treated in a
similar way.
\end{proof}

\section{Numerical experiments}
\label{simul}

Here, the following model is simulated:
$X$ is uniformly distributed on $[0,1]$ and $Y$ given $X=x$
is distributed on $[0,g(x)]$ such that
\begin{equation}
\label{proba}
\PP(Y>y|X=x)=\left(1-\frac{y}{g(x)}\right)^\gamma,
\end{equation}
with $\gamma>0$. This conditional survival distribution function
belongs to the Weibull domain of attraction, with extreme value
index $-\gamma$, see~\cite{EMBR} for a review on this topic.
In the following, three exponents are used $\gamma\in\{1,2,3\}$.
The case $\gamma=1$ corresponds to the situation where $Y$ given $X=x$
is uniformly distributed on $[0,g(x)]$. The larger $\gamma$ is,
the smaller the probability~(\ref{proba}) is, when $y$ is close to
the frontier $g(x)$. 
The frontier function is given by
$$
g(x)=(1/10+\sin(\pi x))\left(11/10-\exp\left(-64(x-1/2)^2\right)/2\right).
$$
The following kernel is chosen 
$$
K(t)=\cos^2(\pi t /2)\mathbf{1}\{t\in[-1,1]\},
$$
and we limit ourselves to first order local polynomials, i.e. $k=1$.
In this case, to fulfill assumption H, one can choose 
$h= c_h n^{-1/2} (\log n)^{1+3\tau/5}$ and $p= c_p n^{1/2} (\log n)^{-1-\tau}$
where $\tau$, $c_h$ and $c_p$ are positive constants. 
In practice, since the choice of $c_h$ and $c_p$ is more important than
the logarithmic factors, we use
$h=4\hat\sigma(X) n^{-1/2}$ and $p=n^{1/2}$.
The multiplicative constants are chosen heuristically.
The dependence with respect to the standard-deviation of $X$
is inspired from the density estimation case. The scale factor 4 was
chosen on the basis of intensive simulations, similarly to~\cite{JMVA2}. 

\noindent The experiment involves four steps:
\begin{itemize}
\item First, $m=500$ replications of a $500-$ sample are simulated.  
\item For each of the $m$ previous set of points, the 
frontier estimator $\hat g_n$ is computed for $k=1$.  
\item The $m$ associated $L_1$ distances to $g$ are evaluated on a grid.  
\item The smallest and largest $L_1$ errors are recorded.
\end{itemize}
\noindent Results are depicted on Figure~\ref{fctludo1}--\ref{fctludo3},
where the best situation
(i.e. the estimation corresponding to the smallest $L_1$ error)
and the worst situation
(i.e. the estimation corresponding to the largest $L_1$ error)
are represented.
Worst situations are obtained when 
no points were simulated at the upper boundary of the support.
To overcome this problem, the normalizing constant $(p+1)$ 
in~(\ref{1-0}) could be modified as in~\cite{JMVA2}, Section~6 to
deal with some particular parametric models of $Y$ given $X=x$.

\section*{Appendix I: Conditional bias of $\widehat{r}_{n}\left(  x\right)  $}

In this appendix, we provide a proof of Proposition~\ref{prop2}.
From $\left(  \ref{1-1}\right)  $, we have
\[
\E\left(  \frac{\widehat{r}_{n}\left(  x\right)  }{r_{n}\left(  x\right)
}-1/\mathcal{X}\right)  =g^{-p}\left(  x\right)  e_{1}^{t}\mathbf{S}_{n}%
^{-1}\mathbf{X}^{t}\mathbf{W}\left[  \E\left(  \mathcal{Z}/\mathcal{X}\right)
-\mathbf{X\beta}\right],
\]
where the term $\E\left(  \mathcal{Z}/\mathcal{X}\right)  -\mathbf{X\beta}$ can
be rewritten as
\[
\left(r_{n}\left(  X_{1}\right)  -\sum_{j=0}^{k}\beta_{j}\left(  X_{1}-x\right)
^{j},...,~r_{n}\left(  X_{n}\right)  -\sum_{j=0}^{k}\beta_{j}\left(
X_{n}-x\right)  ^{j}\right).
\]
Taylor-Lagrange formula with $\beta_{j}=\frac{1}{j!}\frac
{\partial^{j}g^{p}}{\partial x^{j}}\left(  x\right)  $ and $0<\theta<1$ yields
\[
g^{p}\left(  u\right)  =\sum_{j=0}^{k}\beta_{j}\left(  u-x\right)
^{j}+\left(  u-x\right)  ^{k+1}\frac{1}{k+1!}\frac{\partial^{k+1}g^{p}%
}{\partial x^{k+1}}\left(  x+\theta\left(  u-x\right)  \right)  ,
\]
so that, we can derive, for $0<\theta_{i}<1$ depending on $X_{i}$,
the following expansion
\[
r_{n}\left(  X_{i}\right)  -\sum_{j=0}^{k}\beta_{j}\left(  X_{i}-x\right)
^{j}=\left(  X_{i}-x\right)  ^{k+1}\frac{1}{k+1!}\frac{\partial^{k+1}g^{p}%
}{\partial x^{k+1}}\left(  x+\theta_{i}\left(  X_{i}-x\right)  \right)  .
\]
Since $K$ has a bounded support, we have $K_{h}\left(
X_{i}-x\right)  =0$ for $\left\vert X_{i}-x\right\vert >h$. If $\left\vert
X_{i}-x\right\vert \leq h$ and $0<\theta_{i}<1$,
under the conditions $h\rightarrow0$ and $ph\rightarrow
0$, Lemma~\ref{lem11} yields
\begin{align*}
&  \left(  X_{i}-x\right)  ^{j}K_{h}\left(  X_{i}-x\right)  \frac
{\partial^{k+1}g^{p}}{\partial x^{k+1}}\left(  x+\theta_{i}\left(
X_{i}-x\right)  \right)  \\
&  =\left(  X_{i}-x\right)  ^{j}K_{h}\left(  X_{i}-x\right)  \left[
\frac{\partial^{k+1}g^{p}}{\partial x^{k+1}}\left(  x\right)  +\sum
_{j=1}^{k+1}p^{j}g^{p-j}\left(  x\right)  o\left(1\right)  \right].
\end{align*}
Thus, recalling that $S_{n,j}=\sum_{i=1}^{n}\left(  X_{i}-x\right)  ^{j}%
K_{h}\left(  X_{i}-x\right)  $ and $\beta_{j}=\frac{1}{j!}\frac{\partial
^{j}g^{p}}{\partial x^{j}}\left(  x\right)  $,   the $(k+1)$%
-dimensional vector $\mathbf{X}^{t}\mathbf{W}\left(  \E\left(
\mathcal{Z}/\mathcal{X}\right)  -\mathbf{X\beta}\right)  $ can be
rewritten as
\begin{align*}
&  \left[  \sum_{i=1}^{n}\frac{\left(  X_{i}-x\right)  ^{k+j}}{k+1!}%
K_{h}\left(  X_{i}-x\right)  \frac{\partial^{k+1}g^{p}}{\partial x^{k+1}%
}\left(  x+\theta_{i}\left(  X_{i}-x\right)  \right)  \right]  _{j=1,...,k+1}%
\\
&  =\left[  \beta_{k+1}S_{n,k+j}+\frac{1}{k+1!}S_{n,k+j+1}\sum_{j=1}%
^{k+1}p^{j}g^{p-j}\left(  x\right)  o\left(1\right)  \right]
_{j=1,...,k+1}.
\end{align*}
Introducing the vector $c_{n}=\left(  S_{n,k+1},...,S_{n,2k+1}%
\right)  ^{t}$, we obtain%
\[
\mathbf{X}^{t}\mathbf{W}\left(  \E\left(  \mathcal{Z}/\mathcal{X}\right)
-\mathbf{X\beta}\right)  =\beta_{k+1}c_{n}+\frac{1}{k+1!}c_{n}\sum_{j=1}%
^{k+1}p^{j}g^{p-j}\left(  x\right)  o\left(1\right),
\]
and, returning to the bias of $\widehat{r}_{n}\left(  x\right)$,
\begin{equation}
\label{bias}
\E\left(  \frac{\widehat{r}_{n}\left(  x\right)  }{r_{n}\left(  x\right)
}-1/\mathcal{X}\right)  =g^{-p}\left(  x\right)  \beta_{k+1}e_{1}%
^{t}\mathbf{S}_{n}^{-1}c_{n}+\frac{1}{k+1!}e_{1}^{t}\mathbf{S}_{n}^{-1}%
c_{n}\sum_{j=1}^{k+1}p^{j}g^{-j}\left(  x\right)  o\left(1\right).
\end{equation}
Recalling that $\mathbf{S}_{n}=nf\left(  x\right)  \mathbf{H}\mathbf{SH}%
\left[  1+o_{p}\left(  1\right)  \right]  $ with $ \mathbf{H}=diag(1,h,\dots,h^k)$, we have
\[
\mathbf{S}_{n}^{-1}=\frac{1}{nf\left(  x\right)  }\mathbf{H}^{-1}%
\mathbf{S}^{-1}\mathbf{H}^{-1}\left[  1+o_{p}\left(  1\right)  \right].
\]
Besides, introducing the vector $c=\left(  \mu_{k+1},...,\mu_{2k+1}\right)$,
the asymptotic expression of $S_{n,j}$ established in Proposition~\ref{prop1} entails
\[
c_{n}=nh^{k+1}f\left(  x\right)  \mathbf{H}c\left[  1+o_{p}\left(  1\right)
\right]  .
\]
Let us first focus on the first term of the bias expansion~(\ref{bias}):
\begin{align*}
g^{-p}\left(  x\right)  e_{1}^{t}\mathbf{S}_{n}^{-1}\beta_{k+1}c_{n}  &
=g^{-p}\left(  x\right)  \frac{1}{nf\left(  x\right)  }\beta_{k+1}%
nh^{k+1}f\left(  x\right)  e_{1}^{t}\mathbf{H}^{-1}\mathbf{S}^{-1}%
\mathbf{H}^{-1}\mathbf{H}c\left[  1+o_{p}\left(  1\right)  \right] \\
&  =g^{-p}\left(  x\right)  h^{k+1}\beta_{k+1}e_{1}^{t}\mathbf{S}^{-1}c\left[
1+o_{p}\left(  1\right)  \right],
\end{align*}
and using the expression of $\frac{\partial^{k}g^{p}}{\partial x^{k}}$ in
$\left(  \ref{4-1}\right)$, we have
\[
g^{-p}\left(  x\right)  \beta_{k+1}=\sum_{j=1}^{k+1}\frac{p!}{\left(
p-j\right)  !}g^{-j}\left(  x\right)  \phi_{\substack{j\\}}\left(  x\right)
=O\left(  p^{k+1}\right)  ,
\]
leading to
\begin{equation}
\label{bias1}
g^{-p}\left(  x\right)  e_{1}^{t}\mathbf{S}_{n}^{-1}\beta_{k+1}c_{n}=e_{1}%
^{t}\mathbf{S}^{-1}cO_{p}\left(  \left(  hp\right)  ^{k+1}\right)
=O_{p}\left(  \left(  hp\right)  ^{k+1}\right).
\end{equation}
Let us now consider the second term in~(\ref{bias}):
\begin{align*}
\frac{1}{k+1!}e_{1}^{t}\mathbf{S}_{n}^{-1}c_{n}\sum_{j=1}^{k+1}p^{j}%
g^{-j}\left(  x\right)  o\left(1\right) 
&  =e_{1}^{t}\mathbf{S}_{n}^{-1}c_{n}p^{k+1}o\left(1\right) \\
&  =\frac{1}{nf\left(  x\right)  }e_{1}^{t}\mathbf{S}^{-1}\mathbf{H}^{-1}%
c_{n}p^{k+1}o_{p}\left(1\right)
\end{align*}
Expanding $\mathbf{H}^{-1}c_{n}$ we have
\[
\mathbf{H}^{-1}c_{n}=\mathbf{H}^{-1}nh^{k+1}f\left(  x\right)  \mathbf{H}%
c\left[  1+o_{p}\left(  1\right)  \right]  =nh^{k+1}f\left(  x\right)
c\left[  1+o_{p}\left(  1\right)  \right],
\]
which entails
\begin{equation}
\label{bias2}
\frac{1}{nf\left(  x\right)  }e_{1}^{t}\mathbf{S}^{-1}\mathbf{H}^{-1}%
c_{n}p^{k+1}o_{p}\left(1\right)  =e_{1}^{t}\mathbf{S}^{-1}co_{p}\left(
\left(  hp\right)  ^{k+1}\right)  =o_{p}\left(  \left(  hp\right)
^{k+1}\right).
\end{equation}
Collecting (\ref{bias1}) and (\ref{bias2}), 
we obtain the announced result
\[
\E\left(  \frac{\widehat{r}_{n}\left(  x\right)  }{r_{n}\left(  x\right)
}-1/\mathcal{X}\right)  =O_{p}\left(  \left(  hp\right)  ^{k+1}\right).
\]

\section*{Appendix II: Auxiliary results}

We first quote a Bernstein-Fr\'{e}chet inequality adapted to our framework.
\begin{lemma}
\label{lem8}
Let $X_{1},...,X_{n}$ independent centered random variables such that for each
positive integers $i$ and $k$, and for some positive constant $C$, we have
\begin{equation}
\label{condBF}
\E\left\vert X_{i}\right\vert ^{k}\leq k!C^{k-2}\E X_{i}^{2}.
\end{equation}
Then, for every $\varepsilon>0$, we have
$$
\PP\left(  \left\vert \sum_{i=1}^{n}X_{i}\right\vert >\varepsilon
\sqrt{\sum_{i=1}^{n}\E X_{i}^{2}}\right)  \leq2\exp\left\{
-\frac{\varepsilon^{2}}{4+2\varepsilon C/\sqrt{\sum_{i=1}^{n}%
\E X_{i}^{2}}}\right\}.  
$$
\end{lemma}
The proof is standard. Note that condition~(\ref{condBF}) is verified
under the boundedness assumption
$\forall i\geq1$,
$\left\vert X_{i}\right\vert  \leq C$.
In the next lemma, an asymptotic expansion of the estimated regression function $\widehat{r}_{n}\left(  x\right)  =e_{1}^{t}\mathbf{S}_{n}^{-1}X^{t}\mathbf{W}\mathcal{Z}$ is introduced.
\begin{lemma}
\label{lem9}
The estimated regression function $\widehat{r}_{n}\left(  x\right)$
can be rewritten as
$$
\widehat{r}_{n}\left(  x\right)  =\frac{1}{nf\left(  x\right)  }%
\sum_{i=1}^{n}Z_{i}~K_{h}\left(  X_{i}-x\right)  \sum_{j=0}^{k}%
u_{j}\left(  \frac{X_{i}-x}{h}\right)  ^{j}\left[  1+o_{p}\left(  1\right)
\right],  
$$
where $(u_{0,}u_{1},...,u_{k})$ is the first line of the matrix $\mathbf{S}%
^{^{-1}}$.
\end{lemma}

\begin{proof}
It is known from the local polynomial fitting theory that $\widehat{r}%
_{n}\left(  x\right)  =\widehat{\beta}_{0}=e_{1}^{t}\mathbf{S}_{n}^{-1}%
$\textbf{$x$}$^{t}\mathbf{W}\mathcal{Z}$ admits the following asymptotic expression%
\[
\widehat{r}_{n}\left(  x\right) =\frac{1}{nhf\left(  x\right)  }\sum_{i=1}^{n}Z_{i}~K_{0}^{\ast
}\left(  \frac{X_{i}-x}{h}\right)  \left[  1+o_{p}\left(  1\right)  \right],
\]
where%
\[
K_{0}^{\ast}\left(  t\right)  =e_{1}^{t}\mathbf{S}^{-1}\left(  1,t,...,t^{k}%
\right)  K\left(  t\right)
\]
is the so-called \textit{equivalent kernel}, see~\cite{FanGij}. 
The remaining of the proof consists in 
explicitly writing this equivalent kernel. It is worth noticing that
$o_{p}\left(  1\right)  $ depends exclusively of the design $\mathcal{X}$.
\end{proof}

\noindent The following lemma is dedicated to the control of the local
variations of the derivatives of $g^p$, when $p\to\infty$, on a neighborhood of size $h$.
\begin{lemma}
\label{lem11}
Suppose $g$ is a  $C_{k+1}$ function with $k<p$.
If, moreover, $ph\rightarrow0$ and $\left\vert u-v\right\vert \leq h$, then
\[
\frac{\partial^{k+1}g^{p}}{\partial x^{k+1}}\left(  v\right)  =\frac
{\partial^{k+1}g^{p}}{\partial x^{k+1}}\left(  u\right)  +\sum_{j=1}%
^{k+1}\frac{p!}{\left(  p-j\right)  !}g^{p-j}\left(  u\right)  o(1).
\]
\end{lemma}

\begin{proof}
From $\frac{\partial g^{p}}{\partial x}=pg^{p-1}\frac{\partial g}{\partial x}$
and a recurrence argument it is easily checked that
\begin{equation}
\frac{\partial^{k+1}g^{p}}{\partial x^{k+1}}=\sum_{j=1}^{k+1}%
\frac{p!}{\left(  p-j\right)  !}g^{p-j}\phi_j,\label{4-1}%
\end{equation}
where the $\phi_{j}$ are continuous functions. The triangular inequality
entails
\begin{align*}
\left\vert g^{p-j}\left(  u\right)  \phi_{j}\left(  u\right)  -g^{p-j}\left(
v\right)  \phi_{j}\left(  v\right)  \right\vert 
&\leq\left\vert g^{p-j}\left(
u\right)  \right\vert \left\vert \phi_{j}\left(  u\right)  -\phi_{j}\left(
v\right)  \right\vert \\
&+\left\vert \phi_{j}\left(  v\right)  \right\vert
\left\vert g^{p-j}\left(  u\right)  -g^{p-j}\left(  v\right)  \right\vert,
\end{align*}
and, from Lemma~\ref{lem10}, if $ph\rightarrow0$ and $\left\vert
u-v\right\vert \leq h$ we get, for sufficiently large $n$, 
\[
\left\vert g^{p-j}\left(  u\right)  \phi_{j}\left(  u\right)  -g^{p-j}\left(
v\right)  \phi_{j}\left(  v\right)  \right\vert \leq \left\vert
g^{p-j}\left(  u\right)  \right\vert o(1)+\Cu D_{j}\left\vert g^{p-j}\left(
u\right)  \right\vert \left(  p-j\right)  h,
\]
where $D_{j}= \sup_{s\in\left[  u,v\right]  }| \phi_{j}(s)|$. 
Thus,
\begin{align*}
g^{p-j}\left(  v\right)  \phi_{j}\left(  v\right)  &=g^{p-j}\left(  u\right)
\phi_{j}\left(  u\right)  +g^{p-j}\left(  u\right)  (O\left(  ph\right)+o(1))\\
 &=g^{p-j}\left(  u\right)
\phi_{j}\left(  u\right)  +g^{p-j}\left(  u\right)  o(1),
\end{align*}
and replacing in~(\ref{4-1}) yields
\begin{align*}
\frac{\partial^{k+1}g^{p}}{\partial x^{k+1}}\left(  v\right)   &
=\sum_{j=1}^{k+1}\frac{p!}{\left(  p-j\right)  !}g^{p-j}\left(
v\right)  \phi_{\substack{j\\}}\left(  v\right)  \\
&  =\sum_{j=1}^{k+1}\frac{p!}{\left(  p-j\right)  !}g^{p-j}\left(
u\right)  \phi_{\substack{j\\}}\left(  u\right)  +\sum_{j=1}%
^{k+1}\frac{p!}{\left(  p-j\right)  !}g^{p-j}\left(  u\right) o(1) \\
  &
=\frac{\partial^{k+1}g^{p}}{\partial x^{k+1}}\left(  u\right)  +
\sum_{j=1}^{k+1}\frac{p!}{\left(  p-j\right)  !} g^{p-j}\left(  u\right) o(1), 
\end{align*}
and the result is proved. 
\end{proof}

\noindent Let us consider, for $i=1,\dots,n$ the random variables defined by
\[
\xi_{i}=\frac{nh}{pg^{p}\left(  x\right)  }a\left(  X_{i}\right)  \left(
\left(  p+1\right)  Y_{i}^{p}-g^{p}\left(  X_{i}\right)  \right).
\]
The next two lemmas are preparing the application of the
Bernstein-Fr\'echet inequality given in Lemma~\ref{lem8}.
First, it is established that the $\xi_i$ are bounded
random variables. Second, a control of the conditional
variance $\V\left(  \sum\nolimits_{i=1}^{n}\xi_{i}/\mathcal{X}\right)$
is provided. 
\begin{lemma}
\label{lem12}
There exists a positive constant $\Cd$ such that
$\left\vert \xi_{i}\right\vert \leq \Cd$ for all $i=1,\dots,n$.
\end{lemma}

\begin{proof}
Since the kernel $K$ is bounded and has bounded support,
it is easily seen that $a\left(  X_{i}\right)  =0$ if $\left\vert
X_{i}-x\right\vert >h$ and that $a\left(  X_{i}\right)  =O\left(  \frac{1}%
{nh}\right)  $ uniformly in $i$. Noticing that $Y_{i}^{p}\leq$ $g^{p}\left(
X_{i}\right)  $ and using Lemma~\ref{lem10}, we get%
\begin{align}
\left\vert \xi_{i}\right\vert  &  =\frac{nh}{pg^{p}\left(  x\right)
}\left\vert a\left(  X_{i}\right)  \right\vert \left\vert \left(  \left(
p+1\right)  Y_{i}^{p}-g^{p}\left(  X_{i}\right)  \right)  \right\vert
\nonumber\\
&  \leq\frac{nh}{pg^{p}\left(  x\right)  }\left\vert a\left(  X_{i}\right)
\right\vert \left\vert \left(  p+1\right)  g^{p}\left(  X_{i}\right)
\right\vert \label{4-7}\\
&  \leq\frac{nh}{pg^{p}\left(  x\right)  }\left(  p+1\right)  g^{p}\left(
x\right)  \left(  1+ \Cu ph\right)  O\left(  \frac{1}{nh}\right)  \nonumber\\
&  =O\left(  1\right)  \left(  1+O\left(  ph\right)  \right),  \nonumber
\end{align}
and the result is proved.
\end{proof}

\begin{lemma}
\label{lem13}
There exists a positive constant $\Ct$ such that%
\begin{equation}
 \label{4-5}
\frac{nh}{p}\left/\V\left(  \sum_{i=1}^{n}\xi_{i}/\mathcal{X}\right)  \right. 
=\Ct\left[  1+o_{p}\left(  1\right)  \right], 
\end{equation}
or equivalently,
\begin{equation}
 \label{4-6}%
\frac{nh}{p}\left/\sqrt{\V\left(  \sum_{i=1}^{n}\xi_{i}/\mathcal{X}
\right)  }  \right. =\sqrt{\frac{nh}{p}}\sqrt{\Ct}\left[  1+o_{p}\left(  1\right)
\right]. 
\end{equation}
\end{lemma}

\begin{proof}
Recalling that
\[
\V\left(  \left(  p+1\right)  Y_{i}^{p}-g^{p}\left(  X_{i}\right)
/\mathcal{X}\right)  =\V\left(  Z_{i}/\mathcal{X}\right)  =\frac{p^{2}}%
{2p+1}g^{2p}\left(  X_{i}\right),
\]
we can write
\begin{align*}
&\V\left(  \sum_{i=1}^{n}\xi_{i}/\mathcal{X}\right)  \\ 
&=\frac{(nh)^{2}}{g^{2p}\left(  x\right)  }\frac{1}{2p+1}\sum_{i=1}%
^{n}a^{2}\left(  X_{i}\right)  g^{2p}\left(  X_{i}\right)  \\
&  =\frac{h^{2}}{g^{2p}\left(  x\right)  }\frac{1}{2p+1}\frac{1}{f^{2}\left(
x\right)  }\sum_{i=1}^{n}K_{h}^{2}\left(  X_{i}-x\right)  \left[
\sum_{j=0}^{k}u_{j}\left(  \frac{X_{i}-x}{h}\right)  ^{j}\right]  ^{2}%
g^{2p}\left(  X_{i}\right)  \\
&  =\frac{h^{2}}{g^{2p}\left(  x\right)  }\frac{1}{2p+1}\frac{1}{f^{2}\left(
x\right)  }\sum_{j,l=0}^{k}u_{j}u_{l}\sum_{i=1}^{n}K_{h}^{2}\left(
X_{i}-x\right)  \left(  \frac{X_{i}-x}{h}\right)  ^{j+l}g^{2p}\left(
X_{i}\right)  \\
&  =\frac{h^{2}}{g^{2p}\left(  x\right)  }\frac{1}{2p+1}\frac{1}{f^{2}\left(
x\right)  }\sum_{j,l=0}^{k}u_{j}u_{l}\frac{1}{h^{j+l}}S_{n,j+l}^{\ast}.
\end{align*}
Now, substituting the asymptotic expression for $S_{n,j}^{\ast}$ into the
above expression yields
\[
\V\left(  \sum_{i=1}^{n}\xi_{i}/\mathcal{X}\right)  =\frac{nh}%
{2p+1}\frac{1}{f\left(  x\right)  }\sum_{j,l=0}^{k}u_{j}u_{l}\nu_{j+l}\left[
1+o_{p}\left(  1\right)  \right],
\]
and the parts $\left(  \ref{4-5}\right)  ~$and $\left(  \ref{4-6}\right)  $ of
this lemma follow.
\end{proof}

\noindent
The next two lemmas are the key tools to prove Theorem~\ref{th4}.
Lemma~\ref{lem14} is mainly a consequence of the 
Bernstein-Fr\'{e}chet inequality given in Lemma~\ref{lem8}.
Lemma~\ref{lem15} is dedicated to the control of the random variable
$\Delta_n$ introduced in~(\ref{defdelta}).
\begin{lemma}
\label{lem14}
There exists a positive constant $\Cq$ such that for every $\varepsilon>0$,
\[
\PP\left(  \left\vert \widehat{r}_{n}\left(  x\right)  -\E\left(  \widehat{r}%
_{n}\left(  x\right)  /\mathcal{X}\right)  \right\vert \geq\varepsilon
r_{n}\left(  x\right)  /\mathcal{X}\right)  \leq2\exp\left\{  -\Cq\frac
{nh}{p}\varepsilon^{2}\left[  1+o_{p}\left(  1\right)  \right]  \right\},
\]
where the sequence $\left[  1+o_{p}\left(  1\right)  \right]  $ depends
exclusively on the design $\mathcal{X}$ .
\end{lemma}

\begin{proof}
Following the asymptotic expression of $\widehat{r}_{n}\left(  x\right)  $
in Lemma~\ref{lem9}, we can write%
\begin{align*}
&  \PP\left(  \left\vert \widehat{r}_{n}\left(  x\right)  -\E\left(  \widehat
{r}_{n}\left(  x\right)  /\mathcal{X}\right)  \right\vert \geq\varepsilon
r_{n}\left(  x\right)  /\mathcal{X}\right) \\
&  =\PP\left(  \left\vert \sum_{i=1}^{n}a\left(  X_{i}\right)  \left(
Z_{i}-\E\left(  Z_{i}/\mathcal{X}\right)  \right)  \right\vert \left[
1+o_{p}\left(  1\right)  \right]  \geq\varepsilon r_{n}\left(  x\right)
/\mathcal{X}\right) \\
&  =\PP\left(  \left\vert \sum_{i=1}^{n}a\left(  X_{i}\right)  \left(
\left(  p+1\right)  Y_{i}^{p}-g^{p}\left(  X_{i}\right)  \right)  \right\vert
\geq\left[  1+o_{p}\left(  1\right)  \right]  \varepsilon g^{p}\left(
x\right)  /\mathcal{X}\right).
\end{align*}
It is worth noticing that, conditionally to $\mathcal{X}$, the
sequence $o_{p}\left(  1\right)  $ can be seen as a deterministic sequence
converging to $0$. We now introduce the bounded variables $\xi_{i}$ 
(see Lemma~\ref{lem12}).
In accordance with the Bernstein-Fr\'{e}chet 
inequality given in Lemma~\ref{lem8}, and with the expressions~(\ref{4-5}) and~(\ref{4-6})
  in Lemma~\ref{lem13}, we write%
\begin{align*}
&  \PP\left(  \left\vert \widehat{r}_{n}\left(  x\right)  -\E\left(  \widehat
{r}_{n}\left(  x\right)  /\mathcal{X}\right)  \right\vert \geq\varepsilon
r_{n}\left(  x\right)  /\mathcal{X}\right) \\
&= \PP\left(  \left\vert
\sum\nolimits_{i=1}^{n}\xi_{i}\right\vert \geq\left[  1+o_{p}\left(  1\right)
\right]  \varepsilon\frac{nh}{p}/\mathcal{X}\right)  \\
&  =\PP\left(  \left\vert \sum_{i=1}^{n}\xi_{i}\right\vert
\geq\varepsilon\left[  1+o_{p}\left(  1\right)  \right]  \frac{nh}%
{p\sqrt{\V\left(  \sum_{i=1}^{n}\xi_{i}/x\right)  }}\sqrt{\V\left(
\sum\nolimits_{i=1}^{n}\xi_{i}/x\right)  }/\mathcal{X}\right)  \\
&  \leq2\exp\left\{  -\frac{\left(  \varepsilon\left[  1+o_{p}\left(
1\right)  \right]  \frac{nh}{p\sqrt{\V\left(  \sum_{i=1}^{n}\xi
_{i}/\mathcal{X}\right)  }}\right)  ^{2}}{4+2\varepsilon\left[  1+o_{p}\left(
1\right)  \right]  \frac{nh}{p\sqrt{\V\left(  \sum_{i=1}^{n}\xi
_{i}/\mathcal{X}\right)  }}\Cd/\sqrt{\V\left(  \sum_{i=1}^{n}%
\xi_{i}/\mathcal{X}\right)  }}\right\}  \\
&  =2\exp\left\{  -\frac{\left(  \varepsilon\sqrt{\frac{nh}{p}}\sqrt{\Ct%
}\left[  1+o_{p}\left(  1\right)  \right]  \right)  ^{2}}{4+\Cd%
\varepsilon\left[  1+o_{p}\left(  1\right)  \right]  \frac{nh}{p}/\V\left(
\sum_{i=1}^{n}\xi_{i}/\mathcal{X}\right)  }\right\}  \\
&  =2\exp\left\{  -\frac{\varepsilon^{2}\frac{nh}{p}^{{}}\Ct\left[
1+o_{p}\left(  1\right)  \right]  }{4+\Cd \Ct\varepsilon\left[
1+o_{p}\left(  1\right)  \right]  }\right\} \\
& \leq2\exp\left\{  -\Cq\frac
{nh}{p}\varepsilon^{2}\left[  1+o_{p}\left(  1\right)  \right]  \right\},
\end{align*}
and the conclusion follows.
\end{proof}

\begin{lemma}
\label{lem15}
The random variable $\Delta_n$ is bounded conditionally to $\mathcal{X}$, which
means that there exists a positive constant, depending on the design, $\Cc\left(  \mathcal{X}\right) $ such that  $\Delta_n\leq \Cc\left(  \mathcal{X}\right)  $.
\end{lemma}
\begin{proof}
From inequality (\ref{4-7}), we have
\begin{align*}
\left\vert \widehat{r}_{n}\left(  x\right)  \right\vert  &  =\left\vert
e_{1}^{t}\mathbf{S}_{n}^{-1}\mathbf{X}^{t}\mathbf{W}\mathcal{Z}\right\vert \\
&  =\left\vert \frac{1}{nf\left(  x\right)  }\sum\nolimits_{i=1}^{n}%
K_{h}\left(  X_{i}-x\right)  \sum_{j=0}^{k}u_{j}\left(  \frac{X_{i}-x}%
{h}\right)  ^{j}Z_{i}\right\vert \left[  1+o_{p}\left(  1\right)  \right]  \\
&  \leq\left(  \sum_{i=1}^{n}\left\vert a\left(  X_{i}\right)
\right\vert \left(  p+1\right)  g^{p}\left(  X_{i}\right)  \right)  \left[
1+o_{p}\left(  1\right)  \right]  \\
&  =\Cu \frac{p}{h}g^{p}\left(  x\right)  \left[  1+o_{p}\left(  1\right)
\right]  \frac{1}{n}card\left\{  i:\left\vert X_{i}-x\right\vert <h\right\}.
\end{align*}
Then, the strong law of large numbers entails
\[
\left\vert \widehat{r}_{n}\left(  x\right)  \right\vert \leq r_{n}\left(
x\right)  \Cu\frac{p}{h}\left[  \PP\left(  \left\vert X-x\right\vert
<h\right)  \right]  \left[  1+o_{p}\left(  1\right)  \right],
\]
and from the continuity of the density $f$, we have 
$$
\frac{1}{2h}  \PP\left(
\left\vert X-x\right\vert <h\right)    =f\left(  x\right)  \left[
1+o\left(  1\right)  \right].
$$
Consequently,
\[
\left\vert \frac{\widehat{r}_{n}\left(  x\right)  }{r_{n}\left(  x\right)
}\right\vert <2\Cu pf\left(  x\right)  \left[  1+o_{p}\left(  1\right)
\right], 
\]
with $o_{p}\left(  1\right)$ depending on the design $\mathcal{X}$.
We thus write%
\begin{equation}
 \label{4-2}%
\frac{1}{p}\left\vert w_{n}\left(  x\right)  \right\vert =\frac{1}%
{p}\left\vert \frac{\widehat{r}_{n}\left(  x\right)  }{r_{n}\left(  x\right)
}-1\right\vert \leq C\left(  \mathcal{X}\right), 
\end{equation}
where $C\left(  \mathcal{X}\right)  $\ is a positive constant under the
conditioning by $\mathcal{X}$. As an immediate consequence, we get%
\begin{equation}
\left(  1+w_{n}\left(  x\right)  \right)  ^{1/p}-1=o_{p}\left(  1\right)
.\label{4-3}%
\end{equation}
From $\left(  \ref{4-2}\right)  $ and $\left(  \ref{4-3}\right)  $ it is clear
that $\Delta_n$ is bounded conditionally to $\mathcal{X}$.
\end{proof}

\noindent Finally, we quote two results from~\cite{JMVA2} (Lemma~5
and Lemma~4 respectively).
\begin{lemma}
\label{lem10}
If $ph\rightarrow0$, there exists a positive constant $\Cu$ such
that 
$$
g^{p}\left(  x\right)  \leq g^{p}\left(  y\right)  +\Cu g^{p}\left(
y\right)  ph
$$ 
for $\left\vert x-y\right\vert \leq h$.
\end{lemma}

\begin{lemma}
\label{lem16}
There exists a constant $\Cs$ such that $\left\vert u\right\vert <1/2$ entails
$$
\left\vert \left(  1+u\right)  ^{1/p}-1-u/p\right\vert \leq \Cs u^{2}/p.
$$
\end{lemma}

\bibliographystyle{plain}

\begin{thebibliography}{10}

\bibitem{Aragon} Y.~Aragon, A.~Daouia and C.~Thomas-Agnan.
Nonparametric frontier estimation: a conditional quantile-based approach. 
{\em Journal of Econometric Theory}, 21(2):358--389, 2005.

\bibitem{Cazals}
C.~Cazals, J.-P.~Florens and L.~Simar.
Nonparametric frontier estimation: A robust approach. 
{\em Journal of Econometrics}, 106(1):1--25, 2002.

\bibitem{DST}
D.~Deprins, L.~Simar, and H.~Tulkens.
\newblock Measuring labor efficiency in post offices.
\newblock In P.~Pestieau M.~Marchand and H.~Tulkens, editors, {\em The
  Performance of Public Enterprises: Concepts and Measurements}. North Holland
  ed, Amsterdam, 1984.


\bibitem{EMBR} P.~Embrechts, C.~Kl\"uppelberg, and T.~Mikosch.  {\em Modelling extremal events}, Springer, 1997.


\bibitem{FanGij} J.~Fan and I.~Gijbels,
\newblock {\em Local Polynomial Modelling and Applications},
Monographs on Statistics and Applied Probability 66, Chapman \& Hall, London, 1996.

\bibitem{Farrel} M.J. Farrel.
\newblock The measurement of productive efficiency.
\newblock {\em Journal of the Royal Statistical Society A}, 120:253--281, 1957.

\bibitem{ISUPLaurent}
L.~Gardes.
\newblock Estimating the support of a {P}oisson process via the
  {F}aber-{S}hauder basis and extreme values.
\newblock {\em Publications de l'Institut de Statistique de l'Universit\'{e} de
  Paris}, XXXXVI:43--72, 2002.

\bibitem{Geffroy}
J.~Geffroy.
\newblock Sur un probl\`{e}me d'estimation g\'{e}om\'{e}trique.
\newblock {\em Publications de l'Institut de Statistique de l'Universit\'{e} de
  Paris}, XIII:191--210, 1964.

\bibitem{Gijbels2}
I.~Gijbels, E.~Mammen, B.~U. Park, and L.~Simar.
\newblock {On estimation of monotone and concave frontier functions.}
\newblock {\em Journal of the American Statistical Association},
  94(445):220--228, 1999.

\bibitem{Scandi}
S.~Girard and P.~Jacob.
\newblock Extreme values and {H}aar series estimates of point process
  boundaries.
\newblock {\em Scandinavian Journal of Statistics}, 30(2):369--384, 2003.

\bibitem{JSPI}
S.~Girard and P.~Jacob.
\newblock Projection estimates of point processes boundaries.
\newblock {\em Journal of Statistical Planning and Inference}, 116(1):1--15,
  2003.

\bibitem{ESAIM}
S.~Girard and P.~Jacob.
\newblock Extreme values and kernel estimates of point processes boundaries.
\newblock {\em ESAIM: Probability and Statistics}, 8:150--168, 2004.

\bibitem{JMVA2}
S.~Girard and P.~Jacob.
\newblock Frontier estimation via kernel regression on high power-transformed data. 
\newblock {\em Journal of Multivariate Analysis}, 99:403--420, 2008.

\bibitem{JSPI2}
S.~Girard and L.~Menneteau.
\newblock Central limit theorems for smoothed extreme value estimates of point
  processes boundaries.
\newblock {\em Journal of Statistical Planning and Inference}, 135(2):433--460,
  2005.

\bibitem{Hall3}
P.~Hall and B.~U. Park.
\newblock {Bandwidth choice for local polynomial estimation of smooth
  boundaries.}
\newblock {\em Journal of Multivariate Analysis}, 91(2):240--261, 2004.

\bibitem{Hall}
P.~Hall, B.~U. Park, and S.~E. Stern.
\newblock {On polynomial estimators of frontiers and boundaries.}
\newblock {\em Journal of Multivariate Analysis}, 66(1):71--98, 1998.

\bibitem{Har}
W.~H\"ardle, B.~U. Park, and A.~B. Tsybakov.
\newblock Estimation of a non sharp support boundaries.
\newblock {\em Journal of Multivariate Analysis}, 43:205--218, 1995.

\bibitem{JacSuq}
P.~Jacob and P.~Suquet.
\newblock Estimating the edge of a {P}oisson process by orthogonal series.
\newblock {\em Journal of Statistical Planning and Inference}, 46:215--234,
  1995.

\bibitem{Keith}
K.~Knight.
\newblock {Limiting distributions of linear programming estimators.}
\newblock {\em Extremes}, 4(2):87--103, 2001.

\bibitem{KorTsy3}
A.~Korostelev, L.~Simar, and A.~B. Tsybakov.
\newblock Efficient estimation of monotone boundaries.
\newblock {\em The Annals of Statistics}, 23:476--489, 1995.

\bibitem{KorTsy}
A.P. Korostelev and A.B. Tsybakov.
\newblock {\em Minimax theory of image reconstruction}, volume~82 of {\em
  Lecture Notes in Statistics}.
\newblock Springer-Verlag, New-York, 1993.

\bibitem{RupWan}
D.~Ruppert and M.~Wand. 
\newblock Multivariate locally weighted least square regression.
\newblock {\em The Annals of Statistics}, 22:1343--1370, 1994.

\end{thebibliography}

\begin{figure}[p]
\centerline{\includegraphics[height=0.9\textwidth,angle=270]{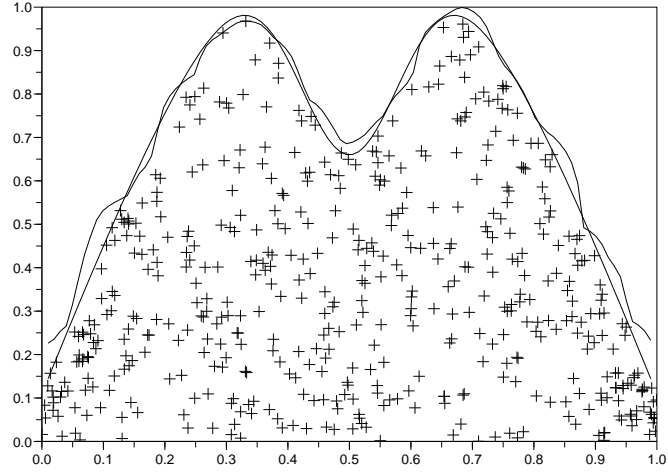}}
\centerline{(a) Best situation} 
\centerline{\includegraphics[height=0.9\textwidth,angle=270]{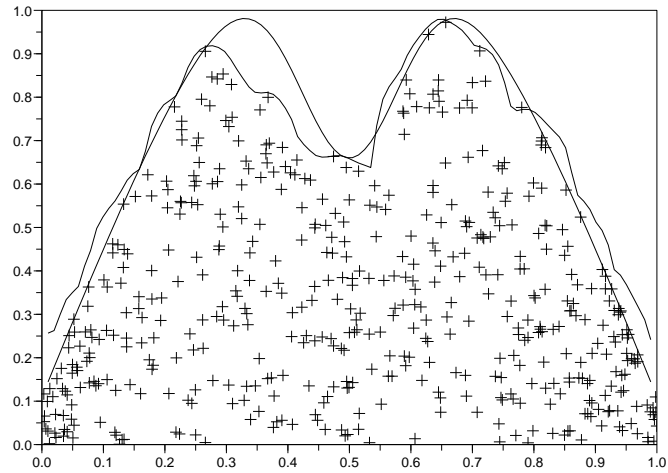}}
\centerline{(b) Worst situation} 
\caption{The frontier $g$ (continuous line) and its estimation (dashed line).
The sample size is $n=500$, $X$ is uniformly distributed on $[0,1]$ and $\gamma=1$.
}
\label{fctludo1}
\end{figure}
\begin{figure}[p]
\centerline{\includegraphics[height=0.9\textwidth,angle=270]{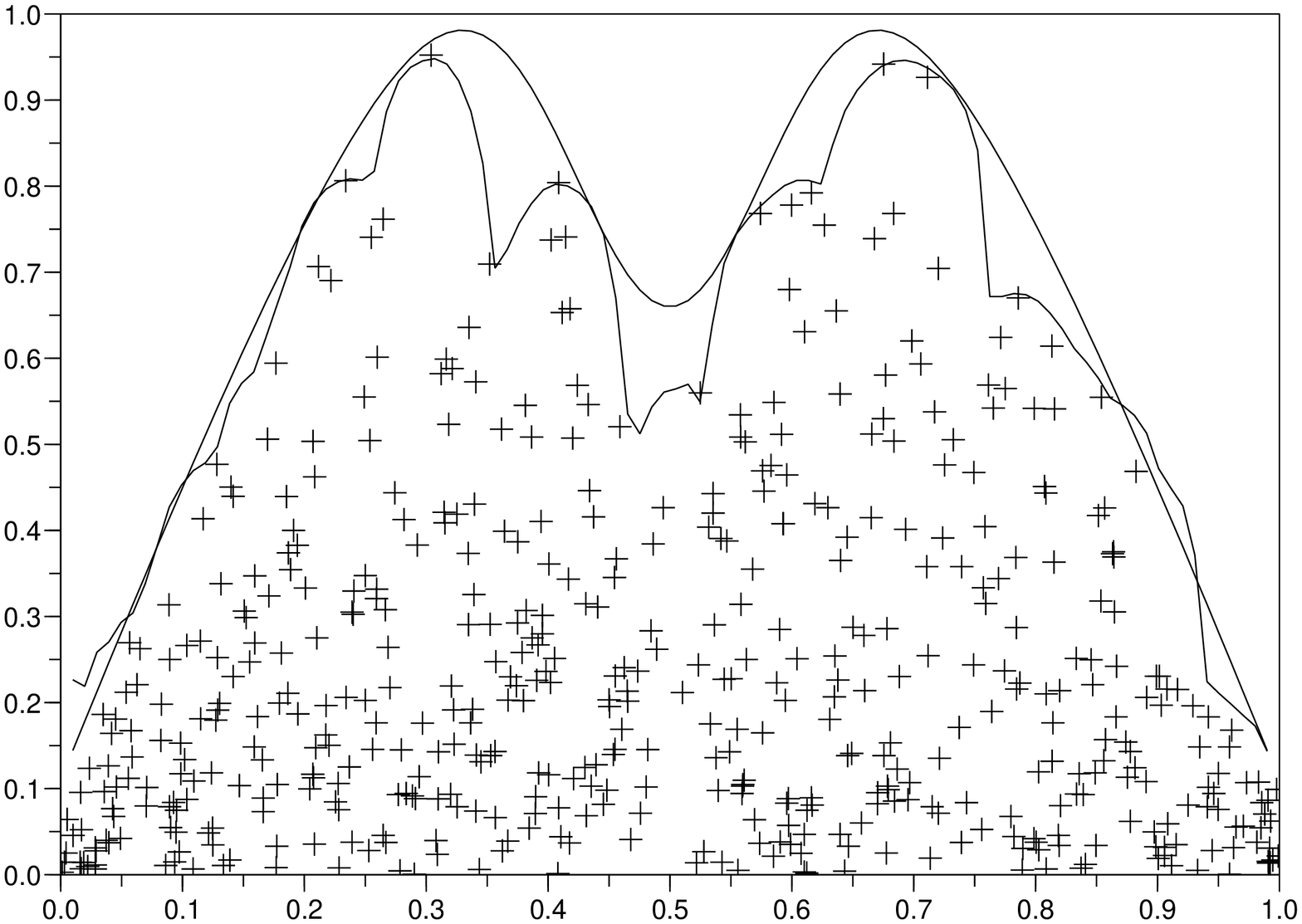}}
\centerline{(a) Best situation} 
\centerline{\includegraphics[height=0.9\textwidth,angle=270]{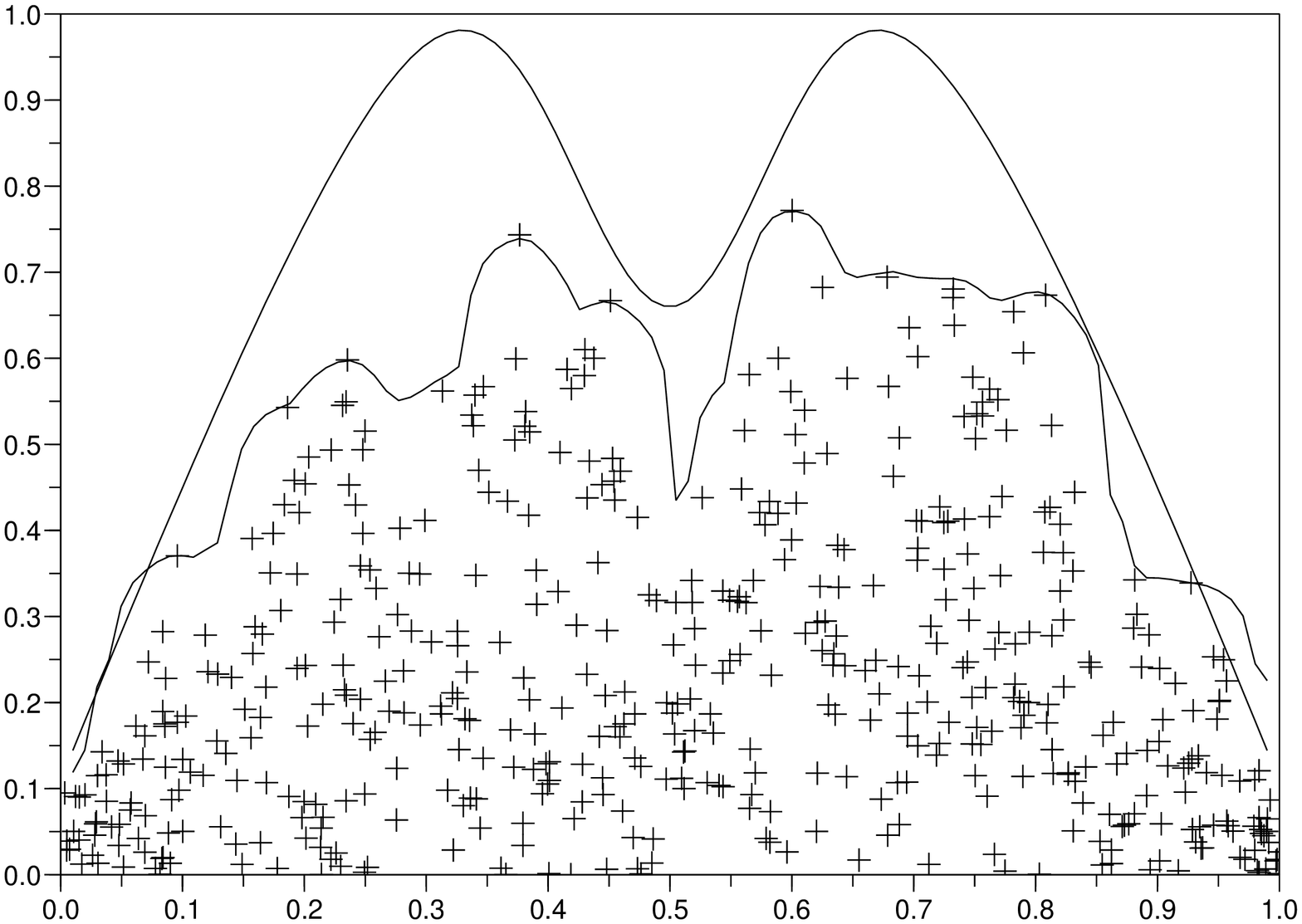}}
\centerline{(b) Worst situation} 
\caption{The frontier $g$ (continuous line) and its estimation (dashed line).
The sample size is $n=500$, $X$ is uniformly distributed on $[0,1]$ and $\gamma=2$.
}
\label{fctludo2}
\end{figure}
\begin{figure}[p]
\centerline{\includegraphics[height=0.9\textwidth,angle=270]{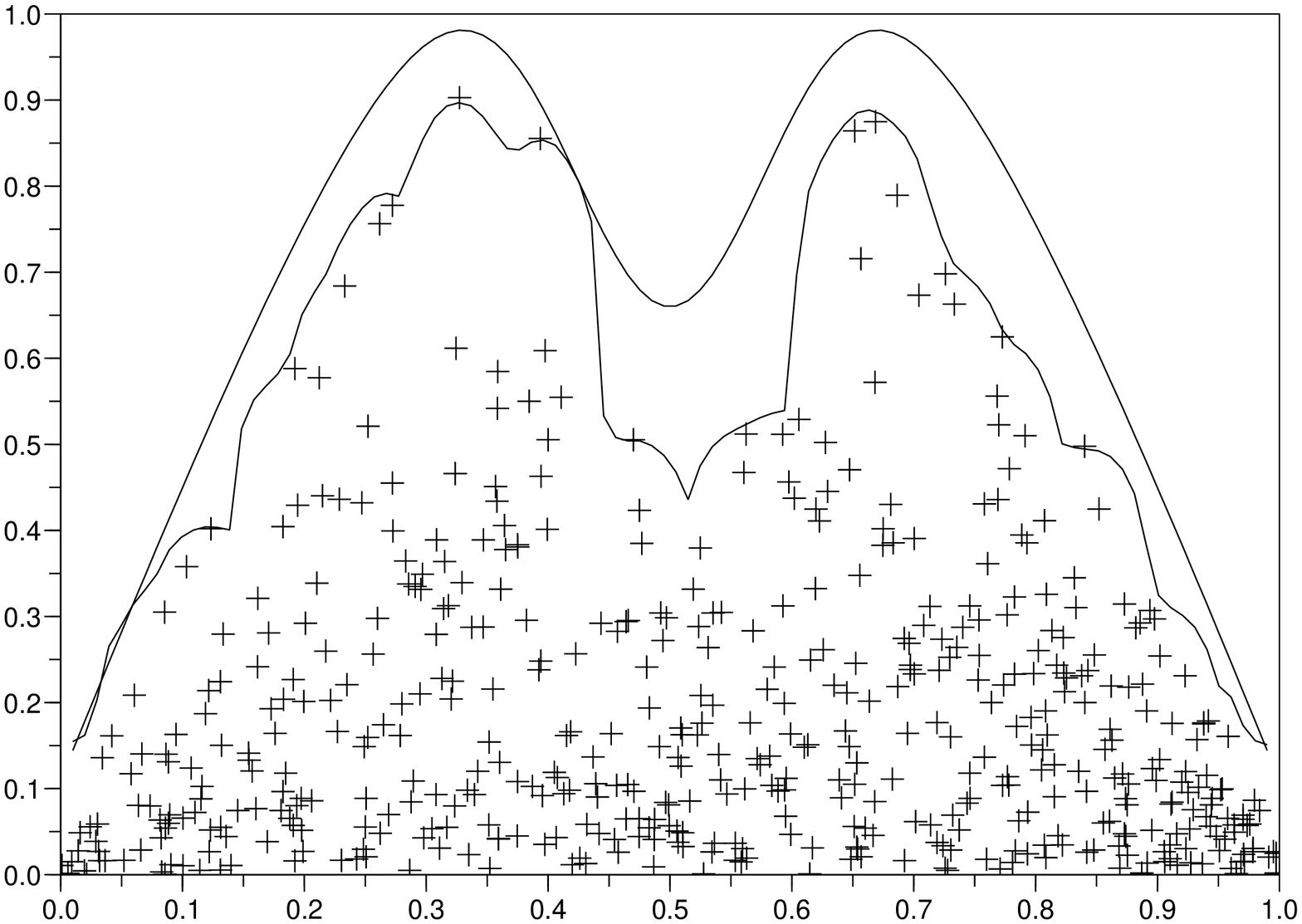}}
\centerline{(a) Best situation} 
\centerline{\includegraphics[height=0.9\textwidth,angle=270]{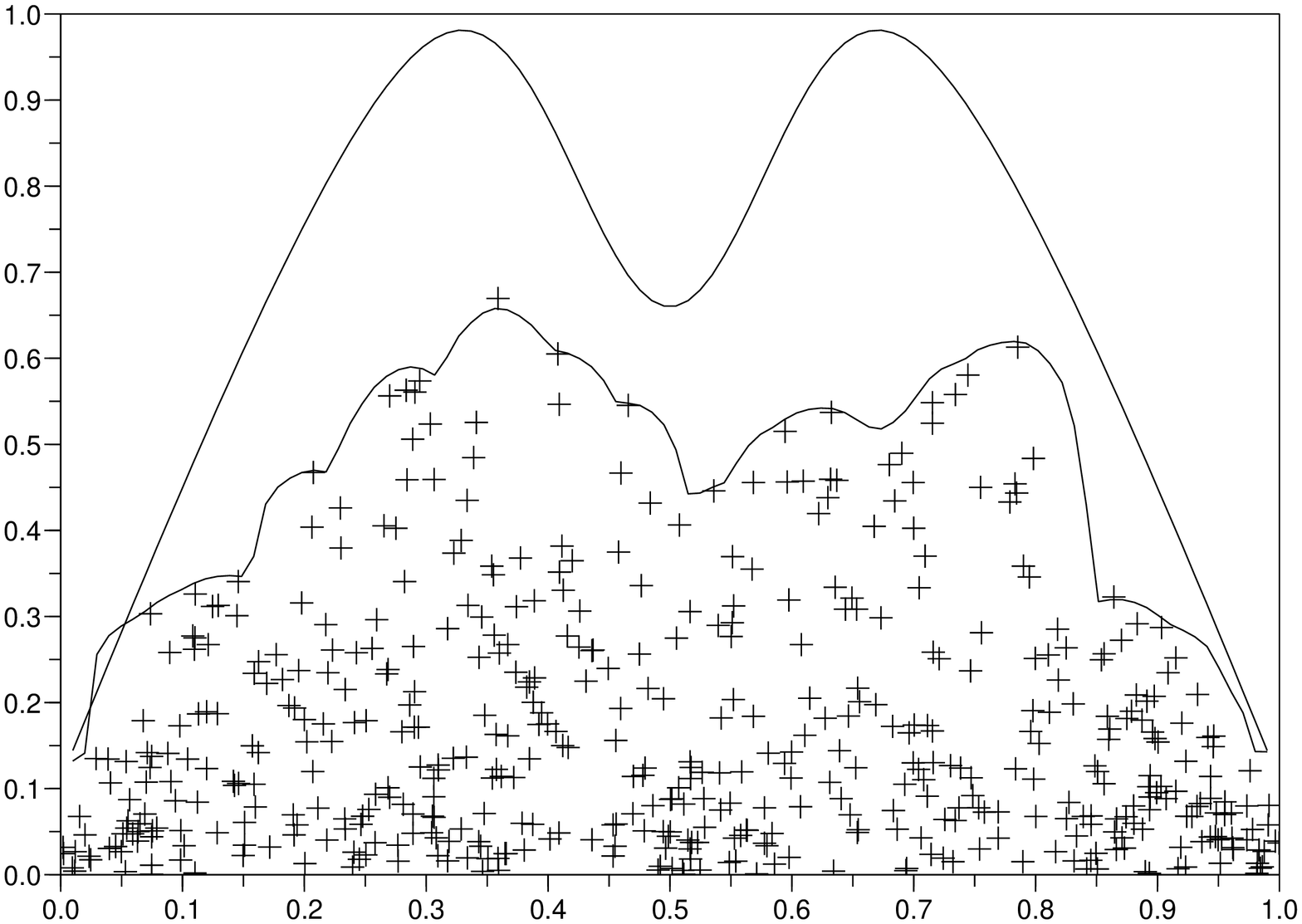}}
\centerline{(b) Worst situation} 
\caption{The frontier $g$ (continuous line) and its estimation (dashed line).
The sample size is $n=500$, $X$ is uniformly distributed on $[0,1]$ and $\gamma=3$.
}
\label{fctludo3}
\end{figure}

\end{document}